\documentclass[8pt,twocolumn,DIV20]{scrartcl}
\pdfoutput=1

\usepackage{titling}
\usepackage{natbib}
\usepackage{amsmath}
\usepackage{amsfonts}
\usepackage{amsthm}
\usepackage{amssymb}
\usepackage[normal, bf]{caption}
\newcommand{\asection}[1]{\paragraph{#1}}
\newcommand{\href}[2]{\footnotesize\texttt{#1}}
\newcommand{\myqed}{}

\usepackage{hyperref}
\usepackage{hyperref}
\hypersetup{
    colorlinks,%
    citecolor=black,%
    filecolor=black,%
    linkcolor=black,%
    urlcolor=black
}

\usepackage{paperpackages}

\begin{document}

\title{Protein Hypernetworks: a Logic Framework for Interaction Dependencies and Perturbation Effects in Protein Networks}
\author{
Johnannes K\"oster%
\protect\thanks{Bioinformatics for High-Throughput Technologies,
Algorithm Engineering, Computer Science 11, TU Dortmund, Germany}
\protect\thanks{Max Planck Institute of Molecular Physiology, Dortmund, Germany}
\and
Eli Zamir\protect\footnotemark[2]\mbox{ }\thanks{to whom correspondence should be addressed}
\and
Sven Rahmann\protect\footnotemark[1]\mbox{ }\protect\footnotemark[3]
}
\date{Inofficial Preprint, \today}
\maketitle

\begin{abstract}
\asection{Motivation:}
Protein interactions are fundamental building blocks of biochemical reaction systems underlying cellular functions.
The complexity and functionality of such systems emerge not from the protein interactions themselves but from the dependencies between these interactions.
Therefore, a comprehensive approach for integrating and using information about such dependencies is required.
\asection{Results:}
We present an approach for endowing protein networks with interaction dependencies using propositional logic, thereby obtaining \emph{protein hypernetworks}.
First we demonstrate how this framework straightforwardly improves the prediction of protein complexes.
Next we show that modeling protein perturbations in hypernetworks, rather than in networks, allows to better infer the functional necessity of proteins for yeast.
Furthermore, hypernetworks improve the prediction of synthetic lethal interactions in yeast, indicating their capability to capture high-order functional relations between proteins.
\asection{Conclusion:}
Protein hypernetworks are a consistent formal framework for modeling dependencies between protein interactions within protein networks.
First applications of protein hypernetworks on the yeast interactome indicate their value for inferring functional features of complex biochemical systems.
\asection{Availability:}
Data and software is publicly available at\\
\href{http://www.rahmannlab.de/research/hypernetworks}{http://www.rahmannlab.de/research/hypernetworks}.
\asection{Contact:}
\href{Eli.Zamir@mpi-dortmund.mpg.de}{Eli.Zamir@mpi-dortmund.mpg.de},\\
\href{Sven.Rahmann@tu-dortmund.de}{Sven.Rahmann@tu-dortmund.de}
\end{abstract}

\section{Introduction}
A fundamental challenge in systems biology is understanding how cellular functions emerge from the collective action of interacting proteins. 
Ultimately such understanding could be achieved through a complete quantitative biochemical description of the system, including the concentrations and spatial distribution of all involved proteins and the kinetic constants of their interactions 
\citep{Hughey:2010fk,Kholodenko:2006uq}.
However, despite the progress in technologies for measuring these parameters in cells, completing such a description for large intracellular biochemical systems is still beyond reach.
In a complementary front, high-throughput protein-protein interaction (PPI) detection techniques, including yeast two-hybrid and mass spectrometry \citep{Walther:2010uq,Parrish:2006fk}, can provide static snapshots of complete interactomes, as demonstrated with several model organisms.
The obtained information is typically modeled as networks -- simple graphs with nodes and edges corresponding to the proteins and their interactions, respectively.
However, such a data structure cannot represent information about how protein interactions depend on each other.

A key mechanism generating interaction dependencies is allosteric regulation, in which a protein undergoes conformational change upon one interaction which affects its other interactions \citep{Laskowski:2009fk}.
Another common type of interaction dependencies is mutual exclusiveness arising from steric hindrance that prevents proteins from binding simultaneously to too close or identical protein domains.
Protein interaction dependencies determine the properties of biochemical systems, and therefore it is essential to comprehensively consider them.
Importantly, vast information about interaction dependencies can be already obtained through database mining, and can be further expanded by high-throughput experimental approaches (see Discussion).
However, a comprehensive approach to integrate this knowledge for getting a better understanding of large biochemical systems is still required.

Recent studies indicate that considering mutual exclusiveness between interactions improves the quality of protein complex prediction in yeast \citep{Ozawa:2010fk,jung2010}.
Here, we further expand and generalize this potential by enabling on one hand the integration of diverse types of interaction dependencies and on the other hand the exploration of different aspects of the system.
We use \emph{propositional logic} to model interaction constraints, and provide a flexible framework for their system-wise representation, called \emph{protein hypernetworks} (Section~\ref{sec:approach}).
Next, we show how to mine hypernetworks for useful information, exemplified here as improving the quality of protein complex prediction (Section~\ref{sec:complexes}).
Furthermore, our approach allows ranking the importance of each protein in a biochemical system based not only on its interactions but also on their dependencies.
We demonstrate that such considerations help predicting which proteins are essential for yeast viability (Section~\ref{sec:functionalimportance}).
Finally, we discuss how our approach synergizes with current efforts to obtain system-level understanding of complex biochemical systems.


\section{Modeling Approach}
\label{sec:approach}

\subsection{Protein Hypernetworks} 
\label{sec:hypernetworks} 

A protein network is commonly described as an undirected graph $(\proteins,\interactions)$ with a vertex $p \in \proteins$ for each protein and an undirected edge $\{p,p'\} \in \interactions$ for each possible interaction.
We first develop an approach for incorporating interaction dependencies into this description, using propositional logic formulas.

The propositional logic $\prop(Q)$ is the set of all propositional logic formulas over the propositions~$Q$ (the atomic units of the logic).
This is the smallest set of formulas such that $q$ itself is a formula for all $q\in Q$
and that is closed under the following operations: For $\phi,\phi' \in \prop(Q)$, all of $\lnot \phi$, $\phi \land \phi'$, $\phi \lor \phi'$, and $\phi \Rightarrow \phi'$ are in $\prop(Q)$ as well.
The operators $\lnot, \land, \lor, \Rightarrow$ have the usual semantics ``not'', ``and'', ``or'', and ``implies'', respectively.
Note that the implication $\phi\Rightarrow\phi'$ is equivalent to $(\lnot\phi \lor \phi')$. 
As propositions $Q$, we use both proteins $\proteins$ and interactions $\interactions$, so $Q := \proteins \cup \interactions$.
A constraint is a formula with a particular structure over these propositions.
\begin{definition}[Constraint]
    A \df{constraint} is a propositional logic formula of the form $q \Rightarrow \psi$
    with $q \in \propositions$ and $\psi \in \prop(\propositions)$.
    With $\constraints(\propositions) \subseteq \prop(\propositions)$ we denote the set of all constraints.
\end{definition}
A constraint $q \Rightarrow \psi$ restricts the satisfiability of $q$ by the satisfiability of $\psi$.
In other words: if $q$ is satisfied, then the same has to hold for $\psi$.
A constraint $q \Rightarrow \psi$ is equivalent to the disjunction $\lnot q \lor \psi$.
We call the disjunct $\lnot q$ the \emph{default} or \emph{inactive} case for the obvious reason that if $q$ is not true, then $\psi$ does not need to be satisfied.
For example (see Fig.~\ref{fig:complexes}a), the dependency of an interaction $i$ on an allosteric effect due to a scaffold interaction $j$ can be formulated by the constraint $i \Rightarrow j$.
Mutual exclusiveness of two interactions $i, j \in I$ can be modelled by the two constraints $i \Rightarrow \lnot j$ and $j \Rightarrow \lnot i$.
The usage of propositional logic allows also to define constraints of higher order:
An interaction $i$ could be either dependent on two scaffold interactions $j_1$ and $j_2$ or compete with an interaction $j_3$, modeled by the constraint $i \Rightarrow ((j_1 \land j_2) \lor \lnot j_3)$.

Now, we can define protein hypernetworks as a set of proteins (nodes) connected by interactions (edges) extended by a set of constraints (dependencies between nodes or edges):

\begin{definition}[Protein Hypernetwork]
Let $\proteins$ and $\interactions$ be sets of proteins and interactions.
Let $C \subseteq \constraints(\propositions)$ be a set of constraints that contains the \df{default constraints} $i \Rightarrow p \land p'$ for each interaction $i = \{p,p'\} \in I$.
Then the triple $(P, I, C)$ is called a \df{protein hypernetwork}.
\end{definition}
Fig.~\ref{fig:complexes}a shows an example protein hypernetwork. While a protein hypernetwork is not a hypergraph (with hyperedges) in the classical sense, the name is appropriate because the constraints describe dependencies between the edges, which could be explicitly transformed into hyperedges, e.g., using minimal network states (cf.\ Sec.~\ref{sec:mining}).

\begin{figure*}[t!]\centering
\includegraphics[width=\textwidth]{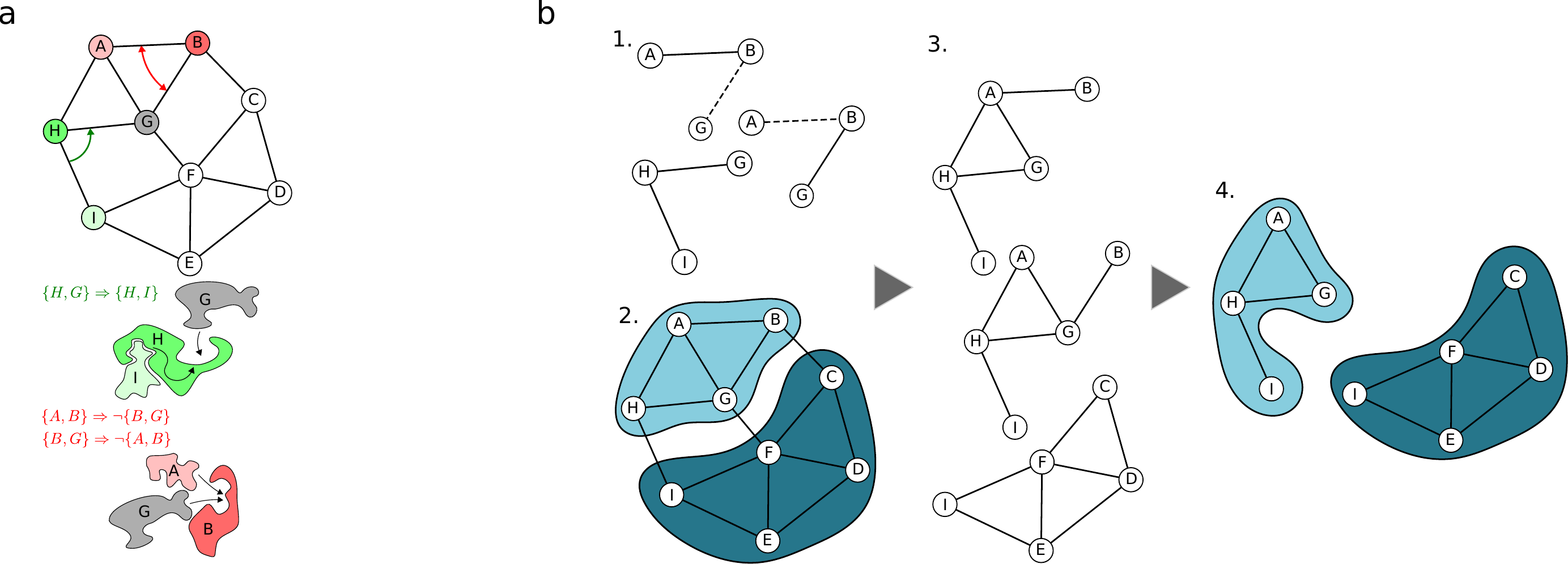}
\caption{\label{fig:complexes}(a) Principle of protein-hypernetwork construction.
A protein network (nodes and black edges) is overlaid with two interaction constraints: mutual exclusive interactions (top arc arrow) and activating allosteric interactions (lower arc arrow).
Propositional logic formulas for these interaction constraints and the molecular mechanism generating them are shown.
The protein hypernetwork is the plain network together with all such propositional logic constraints.
\mbox{ }
(b) Protein complex prediction in four steps (see Sec.~\ref{sec:complexes:algorithm} for algorithmic details):
(1) Computation of minimal network states (Sec.~\ref{sec:mining}),
(2) prediction of initial protein complexes,
(3) computation of simultaneously possible protein subnetworks,
(4) refinement of the predicted complexes.}
\end{figure*}

\subsection{Minimal Network States}
\label{sec:mining}

Following the incorporation of constraints in a protein hypernetwork, we now explain how to sum and propagate their effects in the system.
The key idea is that it is sufficient to examine the implications for each protein or interaction $q \in \propositions$ separately first, and then combine the information in a systematic way.
We formalize this idea by defining sets of \df{minimal network states}.
A minimal network state of~$q$ tells us which other proteins or interactions are \emph{necessary} or \emph{impossible} to occur simultaneously with~$q$.
For each~$q\in\propositions$, we define a \df{minimal network state formula}, for which we then find certain satisfying models, which in turn define minimal network states.

\begin{definition}[Minimal network state formula]\label{def:minstate:formula}
Let $(\proteins,\interactions,C)$ be a protein hypernetwork.
For $q \in \propositions$, the \emph{minimal network state formula} of~$q$ is
\[ \MNS_{(\proteins,\interactions,C)}(q) :=  \MNS(q) := q \land \bigwedge_{c \in C} c \,. \]
\end{definition}

A solution for a propositional logic formula is captured by a \emph{satisfying model} or \emph{interpretation} given by a map $\alpha: \propositions \rightarrow \{0,1\}$ that assigns a truth value to each proposition.
A formula is \emph{satisfiable} if any satisfying model exists.
We assume that $\MNS(q)$ is satisfiable for all $q \in \propositions$, i.e., each single protein or interaction by itself is compatible with all constraints.

For example, consider propositions $Q=\{q_1,q_2\}$ and a formula $\phi = \lnot q_1 \land (q_1 \lor q_2)$.
The only satisfying model is $\propmodel:\nobreak q_1\mapsto\nobreak 0,\, q_2\mapsto\nobreak 1$.
In the protein hypernetworks framework, we interpret a model $\propmodel$ as follows:
A protein or interaction $q$ is said to be \df{possible} in $\propmodel$ iff $\propmodel(q)=1$.
All possible proteins and interactions may (but need not) exist simultaneously (spatially and temporally) in the cell.

There can be many satisfying models for $\MNS(q)$.
Among these, we wish to enumerate all \df{minimally constrained satisfying models} (MCSMs).
A suitable method for finding them is the tableau calculus for propositional logic \citep{smullyan1995}.
In a nutshell, the tableau algorithm decomposes a formula into its parts.
It accumulates conjuncts, branches on disjuncts, and backtracks when a contradiction is encountered.
More details are given in Sec.~\ref{supp:sec:tableau} of the Supplement.
For finding MCSMs, our custom implementation ensures that for each constraint $q\Rightarrow \psi$ (i.e., disjunction $\lnot q \lor \psi$), the default case $\lnot q$ is explored first, and that $\psi$ is expanded only if the constraint is necessarily active, in order to avoid artificially constrained models.

The general problem of deciding whether any given propositional logic formula $\phi$ is satisfiable is NP-complete.
However, $\MNS(q)$ has a special structure: it is a conjunction of a proposition and of (many) constraints.
If all constraints are of a particularly simple structure, we can prove a linear running time; see Sec.~\ref{supp:sec:tableau} in the Supplement.

Each MCSM~$\propmodel$ defines a minimal network state, consisting of both necessary and impossible entities.
The intuition is that the necessary entities~$k$ are simply the ``true'' ones ($\propmodel(k)=1$), and that the impossible entities are those that are \emph{explicitly} forbidden by an active constraint.
\begin{definition}[Minimal Network State]\label{def:minstate}
Let $(\proteins,\interactions,C)$ be a protein hypernetwork and $q \in \propositions$.
Let $\propmodel$ be a MCSM of $\MNS(q)$.
We define sets of \emph{necessary} and \emph{impossible} proteins or interactions, respectively, as
\begin{align*}
 \Nec_{\propmodel}
 &:= \{k \in \propositions \dash  \propmodel(k)=1 \}, \\
 \Imp_{\propmodel}
 &:= \{k \in \propositions \dash  \exists \text{ constraint } (q' \Rightarrow \psi) \in C\\
 &\quad \text{with } \propmodel(q')=1 \text{ and } \psi\land k \text{ is unsatisfiable}.\}.
\end{align*}
The pair $(\Nec_{\propmodel}, \Imp_{\propmodel})$ is called a \df{minimal network state} for $q$ (belonging to the MCSM $\propmodel$).
\end{definition}
For each proposition~$q$, there can be several minimal network states.
We write $M_q$ for the set of all minimal network states for~$q$.
We call $M := M_{\hypernetwork} :=\bigcup_q\, M_q$ the set of all minimal network states for all proteins and interactions.

Now we define a relation \df{clashing}, describing that two minimal network states cannot be combined without producing a conflict.
\begin{definition}[Clashing Minimal Network States]\label{def:clashing}
Two minimal network states $(\Nec,\Imp)$ and $(\Nec',\Imp')$ are \df{clashing}
iff $\Nec \cap \Imp' \neq \emptyset$ or $\Imp \cap \Nec' \neq \emptyset$.
\end{definition}

As we prove in Theorem~\ref{thm:minstate}, in order to know if two proteins or interactions are simultaneously possible, it is sufficient to determine whether any pair of non-clashing minimal network states exists for them.

\begin{theorem}\label{thm:minstate}
Let $(\proteins,\interactions,C)$ be a protein hypernetwork.
Let $q,q' \in \propositions$ be two proteins or interactions, $q \neq q'$.
Assume that there exists a non-clashing pair of minimal network states $(m, m') \in M_{q} \times M_{q'}$.
Then $q$ and $q'$ are possible simultaneously, i.e., the following formula is satisfiable.
\[ \xi := \Big( \bigwedge_{c \in C} c \Big) \land q \land q'. \]%
\end{theorem}%
\begin{proof}\mbox{ }
Let $m = (\Nec, \Imp)\in M_{q}$ and $m' = (\Nec', \Imp')\in M_{q'}$ be non-clashing; 
we show that $\xi$ is satisfiable by defining a satisfiying model $\propmodel$.
Define $\True := \Nec \cup \Nec'$ and $\False := \Imp \cup \Imp'$.
Since $m$ and $m'$ are not clashing, $\True \cap \False = \emptyset$.
Let $\assign{r} := 1$ for $r \in \True$, and $\assign{r}:=0$ otherwise.
We show that~$\propmodel$ satisfies all parts of $\xi$.

The propositions $q$ and $q'$ in $\xi$ are satisfied since $q \in \Nec$ and $q' \in \Nec'$, so $\assign{q}=\assign{q'}=1$.

For each $c^* = (r \Rightarrow \psi)$ in the conjunction $\bigwedge_{c \in C}\, c$, there may appear two cases:
$\assign{r}=0$ or $\assign{r}=1$.
If $\assign{r} = 0$,
then $c$ is satisfied regardless of the satisfaction of $\psi$ because of the implication semantics.
If $\assign{r} = 1$, or equivalently $r \in \True$, then $r\in\Nec$ or $r\in\Nec'$ (or both).
First, consider the case that $r \in \Nec$.
By assumption, $c^*\in C$ is then satisfied in $\bigwedge_{c \in C} c \land q$.
Additionally, it is not clashing with $q'$ because $\True \cap \False = \emptyset$.
Therefore, it is also satisfied in $\xi$.
The case $r\in\Nec'$ is analogous.
\myqed
\end{proof}
Minimal network states are the basis for further inferences on hypernetworks, as we demonstrate in Sections~\ref{sec:complexes} and~\ref{sec:functionalimportance}.
First, however, we show that perturbations can be easily incorporated into the framework.

\subsection{Inclusion of Perturbation Effects}
\label{sec:perturbationeffects}
The protein hypernetwork framework allows to systematically compute consequences of perturbations.
We distinguish between \emph{perturbed} and \emph{affected} proteins or interactions:
A perturbed one is the direct target of an experimental intervention which causes its complete removal from the system (e.g. by gene knock-down for proteins or point mutations for interactions), whereas an affected one is altered due to the propagation of the perturbation in the hypernetwork. 
Assume that proteins $P_\downarrow \subseteq P$ and interactions $I_\downarrow \subseteq I$ are perturbed, and thus removed from the system.
The problem at hand is to compute all affected proteins and interactions.
This is done by recursively removing minimal network states $m=(\Nec, \Imp)$ that necessitate a perturbed or affected entity (protein or interaction) $q \in \Nec$, while counting a protein or interaction as affected once it has no minimal network state left.
Formally, we proceed as follows.

\begin{definition}\label{def:perturbation}
Let $M_q$ be the set of all minimal network states for entity~$q$ (Definition~\ref{def:minstate}),
and let $M\subset\mns$ be any subset of all minimal network states.
For a set of entities $A\subset\propositions$, let
\[ \bar{M}_A := \{(\Nec, \Imp) \in M \,|\,  A \cap \Nec \neq \emptyset \} \]
be the set of minimal network states from~$M$ that become invalid when any entity in $A$ is perturbed.
Let $R(A,M) := M \setminus \bar{M}_A$ be the remaining set of minimal network states.
Let $Q(A,M) := \{q\in\propositions \,|\, M_q \cap R(A,M) = \emptyset \}$ be the set of entities for which no minimal network state is left.

We recursively define a map $\allperturbed$ that maps a set of perturbed entities and a set of minimal network states to the set of affected entities.
Let $\allperturbed: 2^{\propositions} \times 2^{\mns} \rightarrow 2^{\propositions}$ be defined by
\[ \allperturbed(A, M) := \begin{cases}
       \emptyset &\text{if } A = \emptyset,\\
       A \,\cup\, \allperturbed\left(Q(A,M), R(A,M) \right) &\text{otherwise}.
\end{cases}
\]%
\end{definition}

Let $\hypernetwork$ be a protein hypernetwork with perturbations $P_\downarrow \subseteq P$ and $I_\downarrow \subseteq I$ and minimal network states $\mns$.
Then
\[ Q_{\perturbed} := \allperturbed(P_\perturbed \cup I_\perturbed,\, \mns) \]
is the set all affected proteins and interactions.

This provides a module that enables any algorithm that makes predictions based on protein networks to be applied also on a perturbed network, considering the dependencies between interactions.


\section{Result I: Hypernetworks Improve Prediction of Protein Complexes}
\label{sec:complexes}

\subsection{Rationale}
When considering only the interactions between proteins, but not their dependencies, the prediction of protein complexes often relies on identifying dense regions in protein networks~\citep{spirin2003, bader2003,li2005}.
Indeed, algorithms for the prediction of complexes based on plain protein networks $(\proteins,\interactions)$ like MCODE \citep{bader2003} and LCMA \citep{li2005} have been shown to provide reasonable results by detecting such dense regions. 
However, many of the complexes predicted in this way are false positives, since interaction dependencies do not allow their assembly.
Along this line, it was recently shown that consideration of mutual exclusiveness between interactions improve the quality of protein complex prediction \citep{jung2010}.

Here, we first provide a general framework that can build on an arbitrary network-based complex prediction method and ensures that the predicted complexes do not violate arbitrary given interaction constraints within the hypernetwork.
Thus, the framework is much more general than the work by \cite{jung2010}; in practice, however, the main problem is obtaining sufficiently many constraints (see Discussion).
We demonstrate that our framework improves complex prediction on the yeast network in conjunction with the established constraints using the LCM algorithm \citep{li2005} as an example network-based complex prediction method.

\subsection{Algorithm}
\label{sec:complexes:algorithm}

The prediction of protein complexes in hypernetworks consists of four steps, illustrated in Fig.~\ref{fig:complexes}b.
First, for each protein and interaction $q \in P \cup I$, the set of minimal network states $M_q$ is obtained.
Then, with a network based complex prediction algorithm, an initial set of protein complexes is predicted.
Each complex~$c$ is given as a subnetwork $(P_c, I_c)$.

The third step is more complicated.
Let $M_c := \bigcup_{q \in P_c \cup I_c} M_q$ be the set of minimal network states of the complex's entities.
We want to combine the individual states without introducing clashes, as formalized by the following definition.

\begin{definition}[Maximal combination of minimal network states]\label{def:maxcomb}
For a complex $c$, a set~$M \subseteq M_c$ is called a \emph{maximal combination of minimal network states} iff
(1) there exists no clashing pair of minimal network states in~$M$, and
(2) the inclusion of any further minimal network state from $M_c$ would result in a clashing pair.
\end{definition}

All maximal combinations of minimal network states for a given complex~$c$ can be obtained by recursively building a tree of minimal network states to be removed from $M_c$.
The root of the tree is annotated with $M_c$; each other node is annotated with a remaining set $M$.
If $M$ does not contain any pair of clashing states, the node is a leaf, and $M$ is added to the result set of maximal combinations.
Otherwise, we take any $m$ with clashing $m', m'', \dots$ and branch off two children which remove $m$ on the one hand, and remove $m', m'', \dots$ on the other hand.
The tree is explored in a depth-first manner, checking for redundancies in each node.
Let $\maxcombs \subseteq 2^{M_c}$ be the set of all found maximal combinations of minimal network states.
Its cardinality equals the number of non-redundant leaves in the removal tree.
For each maximal combination $M \in \maxcombs$, we generate the corresponding subnetwork of $(P,I)$.

\begin{definition}[Simultaneous Protein Subnetwork]\label{definition:simultaneous_subnetwork}
Let $M \in \maxcombs$ be a maximal combination of minimal network states for complex~$c$.
Let $\proteins_M$ be the set of all necessary proteins and $\interactions_M$ the set of all necessary interactions in $M$, i.e., $\proteins_M := \proteins \cap \bigcup_{(\Nec,\Imp)\in M}\, \Nec$ and $\interactions_M := \interactions \cap \bigcup_{(\Nec,\Imp)\in M}\, \Nec$.
Then $(P_M,I_M)$ is called a \emph{simultaneous protein subnetwork}.
\end{definition}
All proteins and interactions in $(P_M, I_M)$ may exist simultane\-ously in the context of the protein hypernetwork $(P,I,C)$ because the minimal network states in $M$ do not clash with each other.
In comparison to the subnetwork for the network based predicted complex $(P_c, I_c)$, each subnetwork $(P_M, I_M)$ may have lost and gained several interactions or proteins.

Finally, in the fourth step, we perform a network based complex prediction on each simultaneous protein subnetwork $(P_M,P_M)$ again with the same algorithm as during the initial step (as proposed by \citet{jung2010}; thereby it has to be ensured that the network based complex prediction does not produce biased results when performed only on subnetworks).
The proteins and interactions in the new complexes are simultaneously possible.
However, the prediction may miss necessary interactions and proteins outside the initially predicted complex.
Therefore, we force these omitted entities to be contained in the corresponding predicted complex.

\subsection{Experiments}
\label{complex_prediction:experiments}

\begin{table}\centering%
\caption{\label{figure:complexpred}
Application of interaction constraints improves the quality of protein complex prediction.
The procedure described in Sec.~\ref{sec:complexes:algorithm} was performed 
without constraints, with the 458 constraints reported in \citet{jung2010}, and with 100 independent samples of 458 randomly generated constraints (cf.\ Supplement Sec.~\ref{supp:sec:random_constraints}).
(A) Results for the benchmark set of the 55 connected complexes out of the 267 annotated MIPS complexes.
(B) Results for the benchmark set of the 62 connected complexes out of all 1142 MIPS complexes.
}
\begin{tabular}{rll}
\toprule
\multicolumn{1}{@{}l}{\bf(A)}annotated MIPS complexes & precision & recall\\
\midrule
plain (no constraints) & 0.142 & 0.792 \\
458 constraints & 0.206 & 0.792 \\
458 random constr.; mean$\pm$SD & 0.149$\pm$0.005 & 0.782$\pm$0.02\\
\midrule
\multicolumn{1}{@{}l}{\bf(B)}all MIPS complexes& precision & recall\\
\midrule
plain (no constraints) & 0.15 & 0.76 \\
458 constraints & 0.21 & 0.76 \\
\bottomrule
\end{tabular}
\end{table}

To evaluate the refined complex prediction, we use the \emph{Com\-pre\-hen\-sive Yeast Genome Database}  (CYGD; \citet{cygd}) description of the yeast \organism{S.~cerevisiae} interactome (last revision 01-10-2008, 4579 proteins connected by 12567 interactions), which is often being used to benchmark protein complex prediction algorithms \citep{bader2003,li2005,amin2006,jung2010,feng2010}.
This choice ensures consistency with the used collection of interaction dependencies \citep{jung2010}, that was defined on top of CYGD.
Additionally, CYGD contains a gold standard for complex prediction known as the MIPS dataset (1142 known complexes, last revision 18-05-2006), in the following referred to as \emph{CYGD complexes}.
From these, 267 complexes are annotated with their biological function and considered to be reliable \citep{li2005,jung2010,feng2010}.

We exemplarily use the local clique merging algorithm (LCMA, \citet{li2005}, see Supple\-ment Sec.~\ref{supp:sec:lcma}) as the underlying network-based protein complex prediction tool.
Since LCMA cannot predict complexes containing less than three proteins or complexes that are not connected in the underlying protein network, we restrict the benchmark set of the 267~reliable complexes to the 55~connected ones of at least three proteins, and the benchmark set of all 1142~MIPS complexes similarly to 62~complexes.

We created a protein hypernetwork from the yeast protein network by incorporating as constraints 458 pairs of mutually exclusive interactions reported by \citet{jung2010}.
Together, these mutually exclusive interactions constrain 329 interactions, which are 2.7\% of all interactions (we refer to these as the 458 constraints from now, although each mutually exclusive pair of interactions $i,j$ is modelled in fact by two constraints $i \Rightarrow \lnot j$ and $i \Rightarrow \lnot i$).

Our refined predictions are compared to the known CYGD complexes:
Following literature conventions \citep{bader2003,li2005,amin2006,jung2010} to ensure comparability, we consider a predicted complex $c \in \refcomplexes$ to match a CYGD complex $c' \in \cygdcomplexes$ iff
$\sqrt{ (|c \cap c'|^2) / (|c| \cdot |c'|) } \geq \sqrt{0.2} \approx 0.45$.

Let $B:=\cygdcomplexes$ be the benchmark set of CYGD complexes and $P:=\refcomplexes$ be the set of predicted complexes.
By $FP \subseteq P$ we denote the set of false positives, that is predictions that were not found in the benchmark.
By $FN \subseteq B$ we denote the set of false negatives, that is the complexes in the benchmark that were not predicted.
The \emph{recall} and \emph{precision} of a prediction are defined as:
\[ \textit{recall} := \frac{|B| - |FN|}{|B|}
   \quad\text{and}\quad
   \textit{precision} := \frac{|P| - |FP|}{|P|},
\]

Table \ref{figure:complexpred} shows that constraining only 2.7\% of all known interactions already increases the precision while the recall remains constant.
In contrast, applying 458 randomly generated constraints modelling mutually exclusive interactions (see Supplement Sec.~\ref{supp:sec:random_constraints} for details) does not provide an improvement and even reduces the recall. 
Sec.~\ref{supp:sec:gradualcomplex} of the Supplement furthermore shows the development of recall and precision when the true constraints are gradually introduced in a random order.

The prediction of protein complexes in hypernetworks uses network-based complex prediction as an autonomous module.
Therefore, the choice of LCMA as the algorithm for this module here is arbitrary for demonstrating the benefit of using protein hypernetworks.
Any LCMA-analogous algorithm can be plugged similarly into the hypernetwork approach to ensure that the complexes it predicts do not violate constraints, thereby improving the prediction quality.
Indeed, \citet{jung2010} show that this is also true for the MCODE complex prediction algorithm \citep{bader2003}.
Note that the hypernetwork formalism allows to further harness such algorithms to predict changes in protein complexes upon perturbations, as we described in Sec.~\ref{sec:perturbationeffects}.


\begin{figure}
\centering
\includegraphics[width=0.25\textwidth]{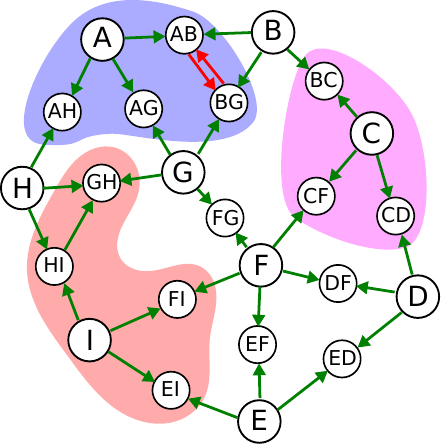}
\caption[Minimal network state graph and PIS]{Minimal network state graph and PIS.
BFS from node $A$ and $I$ results in a higher PIS than from node $C$, because of the competition between interactions $AB$ and $BG$ and the dependency of $GH$ on $HI$, respectively.
For the underlying hypernetwork see Figure \ref{fig:complexes}.}
\label{figure:Gmns}
\end{figure}

\begin{figure*}\centering
\includegraphics[width=0.8\textwidth]{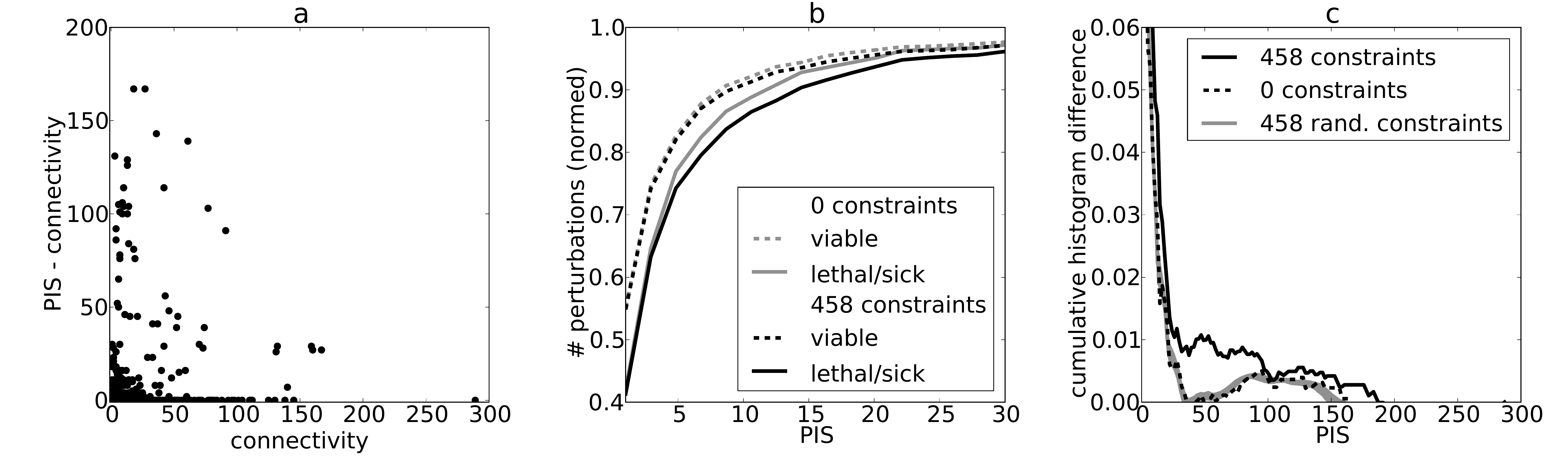}%
\vspace*{-3ex}
\caption{
\label{figure:pis}
Performance of $\piscon$ and $\pisuncon$ (connectivity) in predicting viability of perturbations.
(a) Scatterplot of connectivity against difference PIS minus connectivity; note that always PIS $\geq$ connectivity.
(b) Cumulative distribution function (cdf) of $\piscon$ and $\pisuncon$ for viable and lethal/sick perturbations (axes span region of distinguishable values).
(c) Differences between the cdfs (for details see Supplement, Sec.~\ref{supp:sec:cdf_analysis}) of viable and lethal/sick perturbations, for 0 constraints, 458 constraints, and 458 random constraints. 
For the latter, 100 samples of 458 random constraints were drawn (see Supplement Sec.~\ref{supp:sec:random_constraints}), and the figure shows the area between the mean plus minus one standard deviation.
Note that the $\piscon$ curve is always above the $\pisuncon$ curve, and that the $\pisuncon$ curve agrees well with the randomized curve, indicating the contribution of constraints to a better discrimination between viable and lethal/sick perturbations.}
\end{figure*}


\section{Result II: Hypernetworks Improve Prediction of Protein Functional Importance}
\label{sec:functionalimportance}

\subsection{Perturbation Impact Score}
\label{sec:impact}
Given a plain protein network, the functional importance of each protein is often estimated based on its number of interactions or \df{connectivity} \citep{jeong2001}.
Hypernetworks with their constraints additionally allow to take dependencies between the interactions into account.
We thus propose the \df{perturbation impact score} (PIS) that indicates the amount of changes a perturbation induces in the possible states of a protein network.
First, we define the minimal network state graph that shows the influence of a perturbation (Fig.~ \ref{figure:Gmns}).

\begin{definition}[Minimal network state graph]
Let $\hypernetwork$ be a protein hypernetwork.
Consider a graph $G_{\MNS} := (\propositions,E)$ with directed edges $E$ defined as follows.
For each protein or interaction $q \in \propositions$, consider each minimal network state $(\Nec, \Imp) \in M_q$ and each entity $q' \in \Nec \cup \Imp$; then $E$ consists of all such edges $(q', q)$.
The directed graph $G_{\MNS}$ is called \df{minimal network state graph}.
\end{definition}
An edge $(q', q)$ in $G_{\MNS}$ represents that the perturbation of $q'$ affects the possible configurations of the network around $q$.
On the one hand, if $q'$ is necessary for $q$, then $q$ will become impossible once $q'$ is perturbed.
On the other hand, if $q'$ is mutually exclusive with $q$, the disappearance of $q'$ will also have an effect on the configuration of the network around $q$. 

Now the PIS can be defined for a set of perturbed proteins or interactions.
\begin{definition}[Perturbation Impact Score]
Let $\hypernetwork$ be a protein hypernetwork.
Let $Q_\perturbed\subseteq \propositions$ be the set of perturbed proteins or interactions.
Let $R_\perturbed$ be the set of nodes reachable from $Q_\perturbed$ in the minimal network state graph.
Define a distance function $dist_{Q_\perturbed}: R_\perturbed \rightarrow \mathbb{N}$ such that $dist_{Q_\perturbed}(q)$ is the shortest path length  $G_{\MNS}$ between $q$ and any node in $Q_\perturbed$ (this is well-defined because consideration is restricted to reachable nodes $q$).
The \emph{perturbation impact score} of the set~$Q_\perturbed$ is defined as
\[  PIS_{\hypernetwork}(Q_\perturbed) := \sum_{q \in R_\perturbed}\, dist_{Q_\perturbed}(q).  \]
\end{definition}

The PIS models the idea that a protein or interaction is likely to be more important the further its perturbation propagates through the network.
Of course, $dist_{Q_\perturbed}(q)$ can be computed by a standard breadth-first search (BFS) on $G_{\MNS}$ beginning with the nodes in $Q_\perturbed$.

When computing the perturbation impact score $PIS_{\hypernetwork}(\{p\})$ of a single protein $p$ that does neither appear in a constraint itself nor has a neighbor that does, its score is equal to its connectivity.
Since connectivity was shown to correlate with the functional importance of proteins \citep{jeong2001}, the incorporation of constraints in the PIS can be expected to further enhance this prediction quality. 
As will be shown, PIS allows also to measure the impact of combination of perturbations, thereby to infer functional relations between proteins such as synthetic lethality.

\subsection{Experiments}
\label{sec:pis:experiments}
We evaluated PIS against plain connectivity as estimators for the functional importance of proteins.
Recalling that PIS without constraints equals connectivity, we let $\pisuncon := PIS_{(P,I, \emptyset)}$ be the connectivity, and $\piscon := PIS_{\hypernetwork}$ the score for the constrained CYGD-based yeast protein hypernetwork $\hypernetwork$ as in Sec.~\ref{complex_prediction:experiments}. 
By definition, $\piscon$ is always equal to or higher than connectivity (Fig.~\ref{figure:pis}a).

\begin{figure}\centering
\includegraphics[width=0.7\columnwidth]{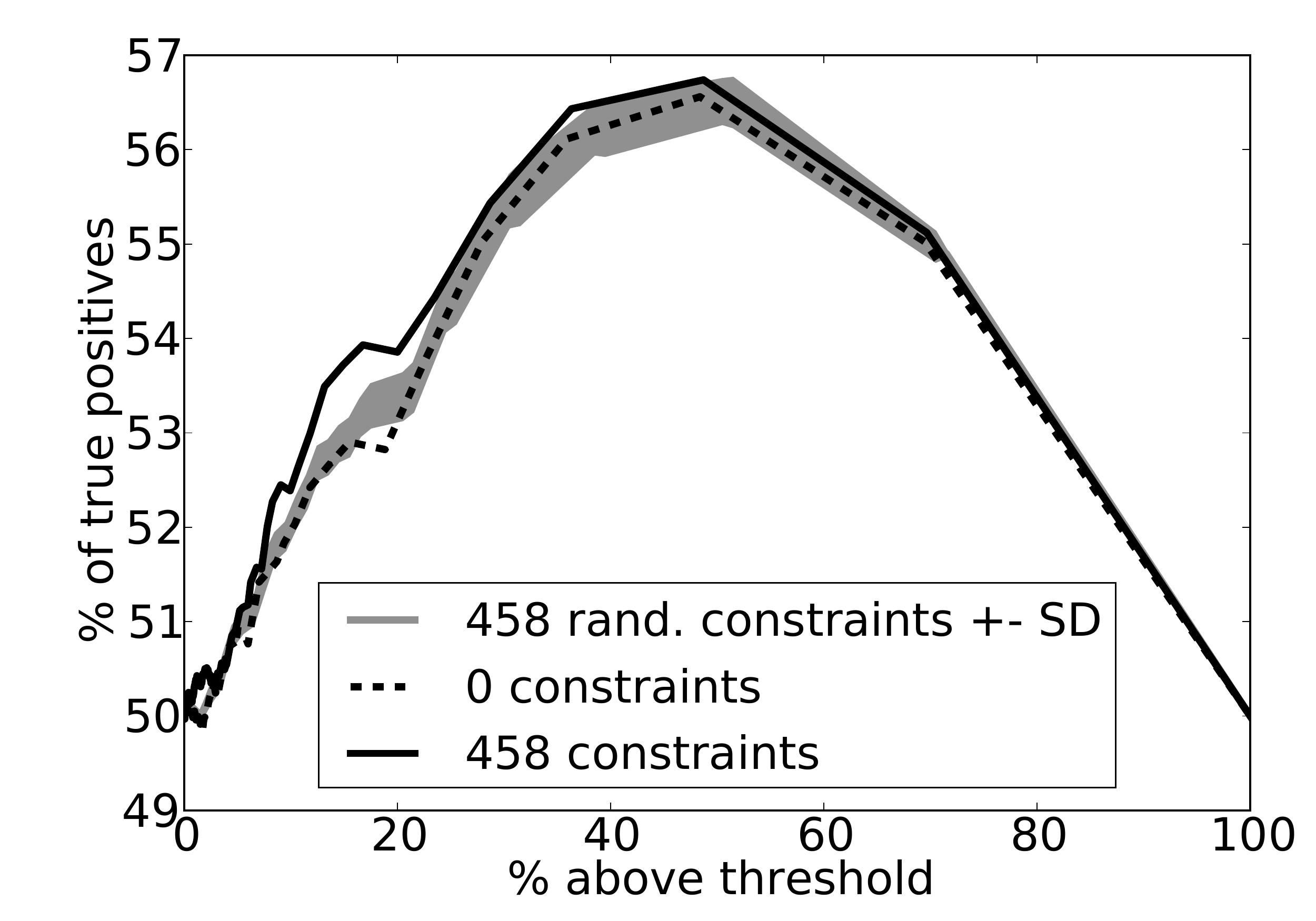}
\caption[Prediction quality depending on threshold]{
\label{figure:threshold_analysis}
Quality of the prediction of lethal/sick and viable perturbations as a function of the threshold used for the classification, for $\pisuncon$, $\piscon$ and 1000 samples of 458 random constraints (see Supplement Sec.~\ref{supp:sec:random_constraints}). 
x-axis: threshold given as percentage of perturbations above a certain PIS or connectivity value;
y-axis: combined ratio of true positives for lethal/sick and viable perturbations
$(TP_\text{lethal}/P_\text{lethal} + TP_\text{viable}/P_\text{viable})/2$.}
\end{figure}

We assume that perturbation of functionally important proteins is more likely to produce sickness or cell death.
Accordingly, for benchmarking, we classified perturbations as \emph{lethal/sick} and \emph{viable} according to the Saccharomyces Genome Database of null mutant phenotypes (1-4-2011, \citet{sgd}; see Supplement, Sec.~\ref{supp:sec:sgd}, for details).
From the distribution of PIS for both of these classes of perturbations (Fig.~\ref{figure:pis}b) it is apparent that proteins resulting lethal/sick null mutants tend to have a higher PIS, regardless of the consideration of constraints. 

Therefore, we investigated more closely the increase in PIS caused by the interaction constraints.
While $7\%$ of the lethal or sick perturbations exhibit an increased PIS upon consideration of constraints, only 2\% of the viable ones do.  
To measure the separation between the classes we subtract the cdf of lethal/sick from that for viable. 
The higher this difference, the better the separation (for details see Supplement, Sec.~\ref{supp:sec:cdf_analysis}).
Fig.~\ref{figure:pis}c shows this measure for 0, 458 and as a mean for 100 samples of 458 random constraints. The application of 458 random constraints does not alter the difference compared to 0 constraints.
In contrast, the application of the true 458 constraints increases the difference -- most obvious for a PIS between 20 and 100 -- and hence improves the discrimination between lethal/sick and viable perturbations.

In general, PIS and plain connectivity agree for most proteins (overlaid points along the horizontal axis in Fig~\ref{figure:pis}a), because only a few constraints are available so far. However, for several proteins the difference is striking.
For example the yeast protein SME1, which is required for mRNA splicing and whose perturbation is lethal to the cell \citep{cygd}, has only $6$ binding partners in the CYGD -- a relatively low connectivity that is not correlated with its biological importance.
With the application of constraints, the PIS of SME1 increases to $111$ and therefore correctly suggests that perturbation of this protein would be lethal.
Counterexamples, where the introduction of constraints wrongly increases the PIS of a viable protein, also exist, however they are a minority, as indicatted by the fact that constraints improve the overall performance of PIS (e.g. Fig.~\ref{figure:pis}c).

To illustrate the use of PIS to predict the functional importance of a protein, we predicted lethal/sick and viable perturbations by systematically applying a threshold $t$ to the PIS. 
If a protein had a PIS of at least $t$ we predicted its perturbation to be lethal/sick, while we predicted it to be viable for a PIS less than $t$.
Fig.~\ref{figure:threshold_analysis} shows the prediction quality for different thresholds, when using $0$ or 458 constraints or 1000 samples of 458 random constraints. 
To ensure the comparability of thresholds $t$, they are expressed as the percentage of proteins reaching or exceeding the PIS or connectivity threshold (e.g., $t=50$ means a certain value of PIS or connectivity such that half of the proteins reach or exceed it).
It can be seen that every non-trivial threshold performs better than the trivial ones that predict no ($t=0$) or all ($t=100$) perturbations to be lethal/sick. 
Further, the application of the true constraints provides an improved prediction (especially below $t=20$) in comparison to random constraints (Fig.~\ref{figure:threshold_analysis}).

\begin{figure}
\centering
\includegraphics[width=\columnwidth]{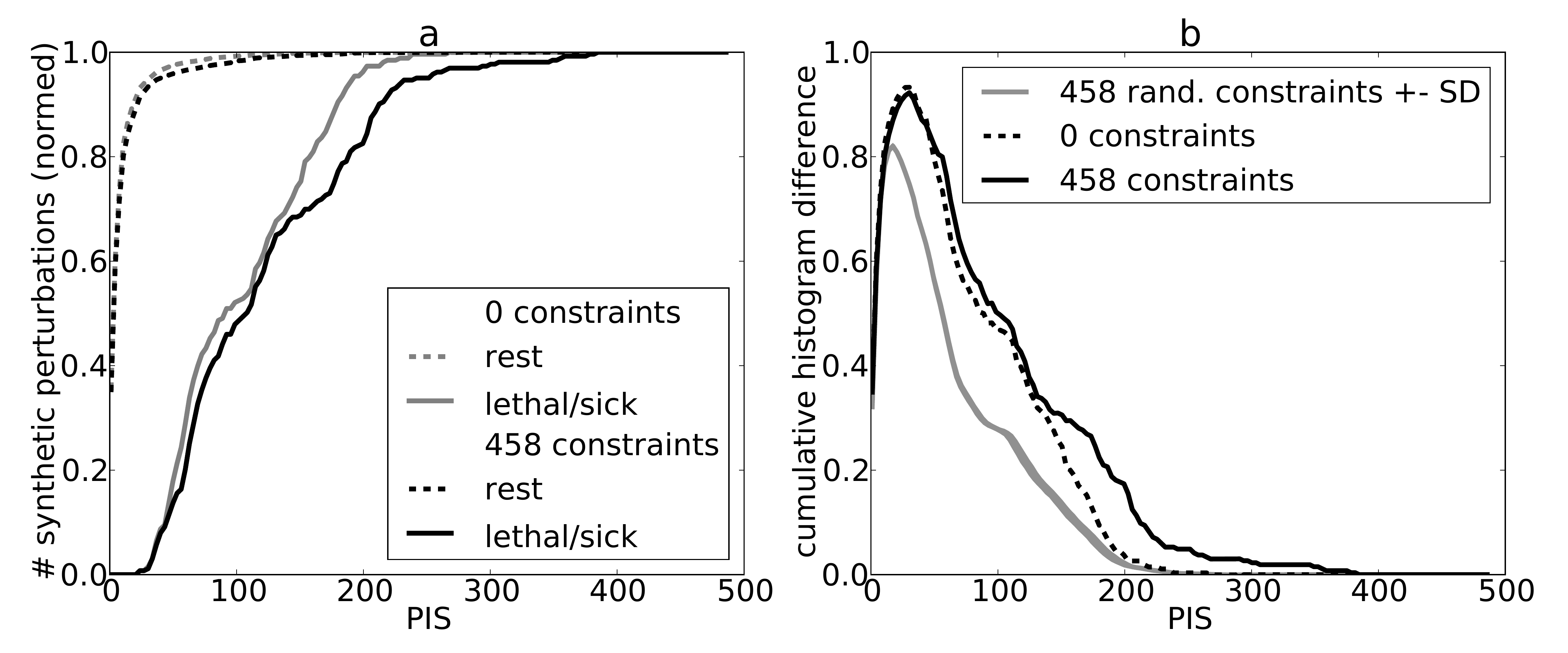}%
\vspace*{-2ex}
\caption[PIS is an indicator for lethal/sick synthetic perturbations]{PIS is an indicator for lethal/sick synthetic perturbations. 
(a) cumulative distributions functions (cdfs) of PIS for synthetic perturbations with 0 and 458 constraints, distinguished between lethal/sick (as defined by \citet{tong2004}) and the rest of synthetic perturbations. 
(b) Differences of cdfs (cf.\ Supplement, Sec.~\ref{supp:sec:cdf_analysis}) of viable and lethal/sick perturbations, for 0 and 458 constraints and 100 samples of 458 random constraints (cf.\ Supplement, Sec.~\ref{supp:sec:random_constraints}). 
Note that for $PIS > 50$ the $\piscon$ curve lies above the $\pisuncon$ curve, indicating the contribution of constraints to a better discrimination between viable and lethal/sick synthetic perturbations.}
\label{figure:synthetic}
\end{figure}

Many null mutations do not affect viability when occurring alone, but become lethal when occurring together with another specific null mutation (i.e. synthetic lethality), indicating functional buffering and relation between the corresponding proteins \citep{tong2001,tong2004}. 
To evaluate if PIS can capture these pair-wise protein relations, we investigated its ability to predict synthetic perturbations by calculating $PIS_{\hypernetwork}(\{p,p'\})$ for every pair of proteins $p,p' \in P$. Note that without constraints, PIS here equals counting the union of neighbors of two proteins.
Fig. \ref{figure:synthetic} shows that PIS provides a striking separation between lethal/sick and viable perturbations (as defined by \citet{tong2004}).
Further, the application of constraints induces a shift of the scores toward higher values, that results in an improved discrimination between viable and lethal/sick synthetic perturbations (Fig.~\ref{figure:synthetic}b). In contrast, using 100 samples of 458 random constraints again does not provide an improvement, and even decreases the discrimination quality.

We conclude that PIS is an improvement over connectivity as a predictive measure for functional importance, as it allows to integrate interaction constraints from hypernetworks.
Since only 2.7\% of interactions are constrained in our experiments, improvements by constraints are naturally small here. We expect the capability of PIS to discriminate between lethal/sick and viable perturbations to further increase as information about additional interaction constraints will become available.


\section{Discussion}
\label{sec:discussion}
The dependencies of protein interactions encode the capability of PPI systems to process information and execute cellular decisions. 
We developed an approach to unfold this dimension of information by incorporating interaction constraints generated by allosteric regulations and competative binding.
On the level of individual proteins, competition between interactions on the same binding domain leads to their complete mutual exclusiveness. 
Similarly, allosteric regulations typically generate all-or-none switches between a non-binding and a binding state \citep{Laskowski:2009fk}. 
Therefore, propositional logic can capture perfectly these fundamental processes, and additionally facilitates their algorithmic integration.
Similarly, protein hypernetworks can incorporate regulations of protein interactions by post-translational modifications (e.g. phosphorylation), as these are often on/off switches describable by propositional logic (e.g. $\{A,B\} \Rightarrow \{B,PO_4\}$ states that B has to be phosphorylated to allow its interaction with A).
The temporal expression and spatial distribution of intracellular proteins, which were shown to be valuable dimensions of information \citep{Han:2004fk,Walther:2010uq}, can also be incorporated into the hypernetwork framework by discretizing them in time (e.g. cell-cycle phases or developmental stages) and space (e.g.\ by sub-cellular compartment or by tissue).

The question addressed in this work is how to use information about interaction dependencies, rather than how to collect it. 
Nevertheless, it should be noted that a significant amount of information about interaction dependencies can already be obtained through curation from literature. Along this line, we are currently developing a text-mining tool to assist the identification of publications reporting interaction dependencies. 
As for future publications, since automatic curation of protein interactions is facilitated by a structured text format \citep{Leitner:2010fk,Ceol:2008fk}, our work motivates its usage to report interaction dependencies (see Supplementary Sec.~\ref{supp:section:representation}). 
Mutual exclusiveness between protein interactions can also be inferred from protein-domain-annotated interactome databases \citep{Ooi:2010fk} or in-silico docking modeling \citep{Wass:2011fk,Mosca:2009fk}. 
Finally, high-throughput quantification of protein interactions at domain resolution and methods for monitoring high-order interactions \citep{Jain:2011fk,Hruby:2011fk,Heinze:2004uq} would provide comprehensive identification of interaction dependencies.
		
Here, we illustrated that even constraining less than 3\% of the interactions in the CYGD is sufficient to improve complex prediction, consistently with previous results \citep{jung2010}.
It is expected that the actual fraction of constrained interactions is much higher, to allow a dynamic and functional yeast interactome.
These additional interaction constraints would be due to allosteric regulations \citep{Laskowski:2009fk}, generation and elimination of binding sites upon protein phosphorylation and dephosphorylation \citep{Seet:2006fk} and more cases of mutually exclusive interactions. 

We proposed a perturbation impact score that provides a measure for a protein's importance within a hypernetwork.
We have shown that this measure provides improvements to the prediction of functionally important proteins compared to the investigation of plain connectivity due to the usage of interaction dependencies as constraints.
As more constraints get reported, the measure should help to rationally design perturbation experiments for network analysis \citep{Zamir:2008fk} and provide mechanistic insights into large PPI systems.

Our data and software for protein hypernetworks are available;
please refer to the Supplement (Sec.~\ref{supp:sec:implementation}) for implementation details.

\bibliographystyle{natbib}
\bibliography{document}

\setcounter{section}{0}
\setcounter{figure}{0}
\setcounter{table}{0}
\setcounter{theorem}{0}

\title{Protein Hypernetworks: a Logic Framework for Interaction Dependencies and Perturbation Effects in Protein Networks (Supplementary Data)}
\author{}
\date{}
\maketitle
\makeatletter \renewcommand{\thefigure}{S\@arabic\c@figure} \renewcommand{\thetable}{S\@arabic\c@table} \renewcommand{\thesection}{S\@arabic\c@section} \renewcommand{\thetheorem}{S\@arabic\c@theorem} \makeatother 

This supplement contains background reference material on the Tableau Algorithm (Sec.~\ref{sec:tableau}), 
additional material on protein complex prediction with hypernetworks (Sec.~\ref{sec:complexpred}), 
details on the generation of random constraints for null models (Sec.~\ref{sec:random_constraints}),
supplementary material on the prediction of protein functional importance with the PIS defined in the main aricle (Sec.~\ref{sec:pis}).
We also provide software implementation details (Sec.~\ref{sec:implementation}), including resource consumption and details on the representation of constraints.


\section{Background: Tableau Algorithm}
\label{sec:tableau}
A suitable method for finding satisfying models is the tableau calculus for propositional logic \citep{smullyan1995}:
For an input formula $\phi$, it generates a deductive tree (the \df{tableau}) of assumptions about $\phi$.
Each assumption $a$ in the tree can be made due to an assumption $a'$ in an ancestral node.
We say that $a'$ \df{results in} $a$, and the generation of $a$ out of $a'$ is called \df{expansion} of $a'$.
The propositional logic tableau algorithm generates satisfying models $\alpha$ for $\phi$.
We write $\propmodel \Vdash \psi$ if $\propmodel$ satisfies a subformula $\psi$.
The tableau algorithm now generates assumptions of the type $\propmodel \Vdash \psi$ with $\psi$ being a subformula of the input formula. That is, a conjunction $\propmodel \Vdash \psi_1 \land \psi_2$ is expanded into $\propmodel  \Vdash \psi_1$ and $\propmodel  \Vdash \psi_2$ on the same path, and a disjunction $\propmodel \Vdash \psi_1 \lor \psi_2$ results in branching into $\propmodel  \Vdash \psi_1$ and $\propmodel  \Vdash \psi_2$ (see Fig.~\ref{fig:tableau}).

Each path from the root to a leaf represents a model $\propmodel$.
If a path does not contain any contradictory assumptions, the model satisfies the input formula.
Implementations of the tableau algorithm explore the tree in a depth-first way, and use backtracking once a contradiction occurs.

\begin{figure}[t]\centering%
    \begin{tikzpicture}
        \treestyle
        \node{$\propmodel \Vdash \lnot A \land (A \lor B)$}
        child {
            node{$\propmodel \Vdash \lnot A$}
            child{
            node{$\propmodel \Vdash A \lor B$}
            child {
                node{$\propmodel \Vdash A \lightning$}
            }
            child {
                node{$\propmodel \Vdash B \checkmark$}
            }}
        };
    \end{tikzpicture}
\caption{
Tableau for the propositional logic formula $\phi = \lnot A \land (A \lor B)$.
The path marked by $\lightning$ does not lead to a satisfying model $\propmodel$, because it contains a contradiction between the assumptions $\propmodel \Vdash \lnot A$ and $\propmodel \Vdash A$.
The path marked by $\checkmark$ is free of contradictions, hence its generated model satisfies $\phi$.}
\label{fig:tableau}
\end{figure}

Different variations of the tableau algorithm exist.
For example, one may be interested only in the decision ``does a satisfying model exist?'', or the task could be to output an (arbitrary) satisfying model (if one exists), or to list all satisfying models.
The latter is the task we face when enumerating minimal network states.
In theory, the tableau algorithm exhibits an exponential worst case complexity, as it operates by complete enumeration of all cases with backtracking.

However, elaborate backtracking strategies can significantly reduce the running time in practice.
Insights into such strategies and implementation details are provided by \citep{zhen_li2008}.
Also, faster heuristics exist, like GSAT \citep{selman1992}, but they are not adequate for the problem, as they do not guarantee a correct and complete answer.

For our purpose, the implementation has to ensure that
\begin{enumerate}
\item for each constraint $q\Rightarrow \psi$ (i.e., disjunction $\lnot q \lor \psi$), the default case $\lnot q$ is explored first, and that $\psi$ is expanded only if the constraint is necessarily active, in order to avoid artificially constrained models;
\item all satisfying models that comply with 1.\ are enumerated in the process.
\end{enumerate}

In our application, the tableau algorithm can be expected to perform acceptably, since we solve only minimal network state formulas with a fixed structure (a conjunction of constraints) and expect most constraints to be of a simple form.
In particular, we expect mostly mutual exclusive interactions, modeled by constraints of the form $i \Rightarrow \lnot j$, and scaffold dependent interactions that can be represented by a constraint of the form $i \Rightarrow j$.
To prove the performance of the tableau algorithm when all constraints are of this form, for a protein hypernetwork $(\proteins,\interactions,C)$ we now show that it will need only $\mathcal{O}(|C|)$ expansions to find a satisfying model.
Since expansions generate the deductive tree, that also limits all tableau operations like backtracking or contradiction tests to be polynomial in $|C|$.

\begin{theorem}
Let $\MNS_{(\proteins,\interactions,C)}(q)$ with $q \in\propositions$ be the minimal network state formula for a protein hypernetwork $(\proteins,\interactions,C)$.
Assume that each constraint in $C$ is of the form $c = (q_1 \Rightarrow \ell)$ with a literal $\ell \in \{q_2, \lnot q_2\}$ and $q_1, q_2\in\propositions$.
Then the tableau algorithm needs at most $\mathcal{O}(|C|)$ expansions to find a satisfying model.
\end{theorem}
\begin{proof}
We show that a constraint that is active cannot be rendered inactive again when assuming that the formula is satisfiable.
Assume that an active constraint $c = (q_1 \Rightarrow \ell)$ is the cause of a conflict, hence $\ell$ contradicts some literal $\ell'$. 
Since we require above that the tableau explores the inactive case first, we know that $\lnot q_1$ caused a contradiction, too. 
We now assume that $\ell$ is removed and we expand $c$ to $\lnot q_1$ again to resolve the contradiction. 
Then, the formula is found to be not satisfiable, because $\lnot q_1$ can either contradict $q$ or another constraint, in which case the argument can be applied recursively.

Now we show that each constraint is expanded at most two times. There are three cases:
(1) The constraint is never activated; then only the inactive case is expanded and the constraint is expanded only once. 
(2) The constraint is activated immediately because $\lnot q_1$ leads to a conflict. 
This needs two expansions. 
(3) The constraint is first inactive and then activated because of a backtracking. 
This needs again two expansions.
Hence, the tableau algorithm needs to perform $1 + 2|C| = \mathcal{O}(|C|)$ expansions.
\qed
\end{proof}

Note that $\MNS_{(\proteins,\interactions,C)}(q)$ with above simple constraints is essentially a Horn formula for which it is known that calculating a satisfying model has polynomial complexity with specialized algorithms \citep{dowling1984}. 
However, proving the complexity of the general tableau algorithm for this case remains useful: 
While we expect most of our constraints to have this simple form, we cannot be sure for all of them.
Hence it is reasonable to provide a computational approach that can handle full propositional logic, but will have comparable complexity to specialized horn formula algorithms in the majority of cases.


\section{Complex Prediction in Hypernetworks}
\label{sec:complexpred}

In this section, we provide more details on protein complex prediction, supplementing Section~\ref{main:complex_prediction:experiments} of the main article.

\subsection{Background: LCMA}
\label{sec:lcma}
\begin{figure}\centering
\includegraphics[width=0.27\textwidth]{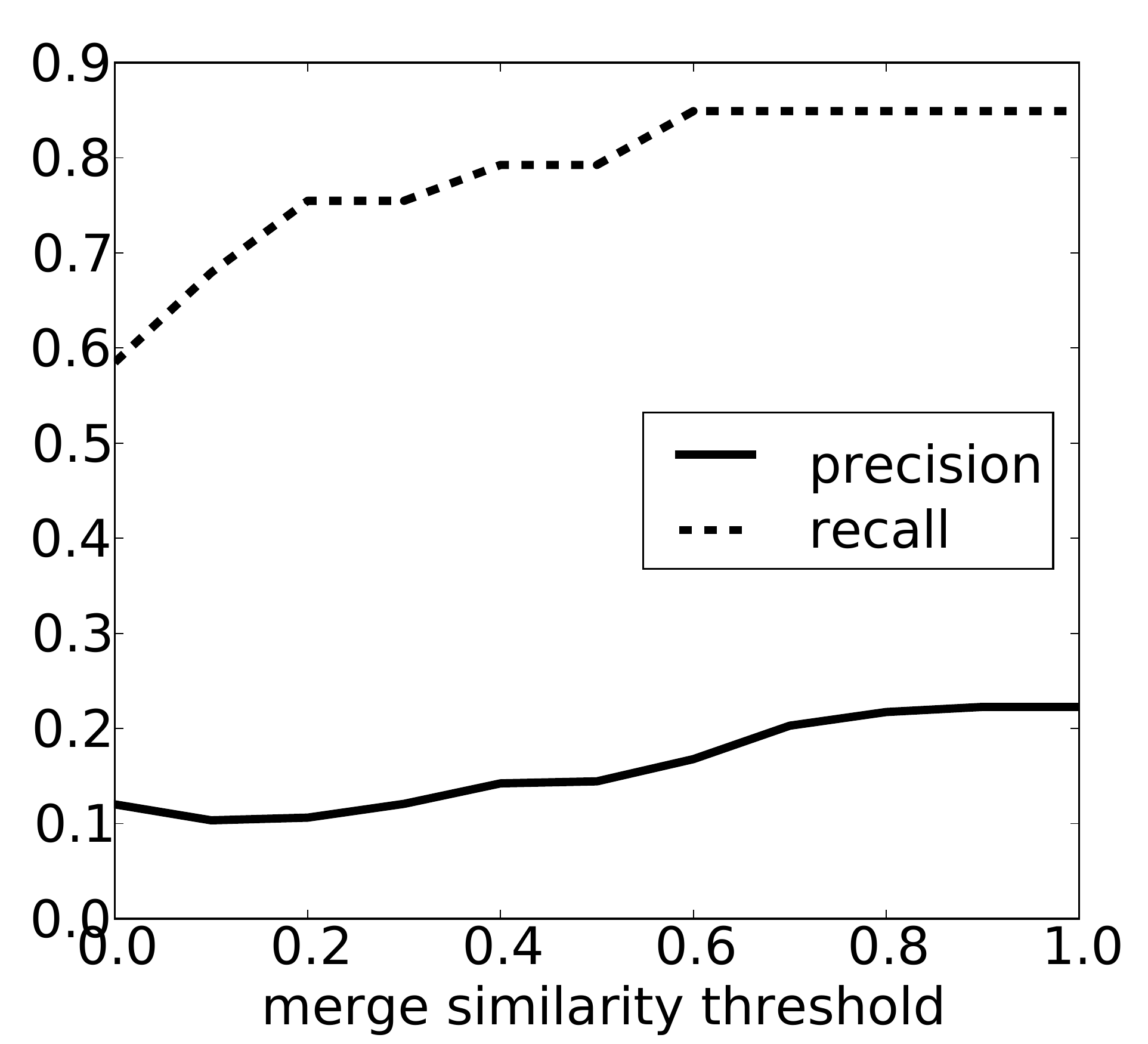}%
\caption{LCMA precision and recall for different merge similariy thresholds~$\omega$ on the CYGD protein network and complexes, as decribed in the main article.}
\label{fig:lcma_thresholds}
\end{figure}

We briefly summarize the steps of the Local Clique Merging Algorithm (LCMA; \cite{li2005}) as an exemplary complex prediction algorithm based on plain networks that can be improved by introducing protein hypernetworks.
In a first step, LCMA finds a set of local cliques; 
second, it iteratively merges those with a significant overlap (given by a \emph{merge similarity threshold}~$\omega$). 
The complex prediction consists of all merged cliques once no further merges happen or average density falls below 95\% of the previous iteration.
While the authors propose $\omega=0$, we found LCMA to perform better with higher thresholds on the plain CYGD \citep{cygd} network and complexes (Fig.~\ref{fig:lcma_thresholds}).
Now, a higher threshold $\omega$ means that less clique merging is performed.
The best choice of $\omega=1.0$ means that the clique merging step is not performed at all because only cliques that overlap by $100\%$ (i.e. that are identical), would be merged. 
Therefore, we chose $\omega = 0.4$ heuristically as a compromise between prediction quality and originally intended behaviour.

\subsection{Effects of Gradually Introducing Constraints on Complex Prediction in Hypernetworks}
\label{sec:gradualcomplex}

\begin{figure*}[t!]\centering
\includegraphics[width=\textwidth]{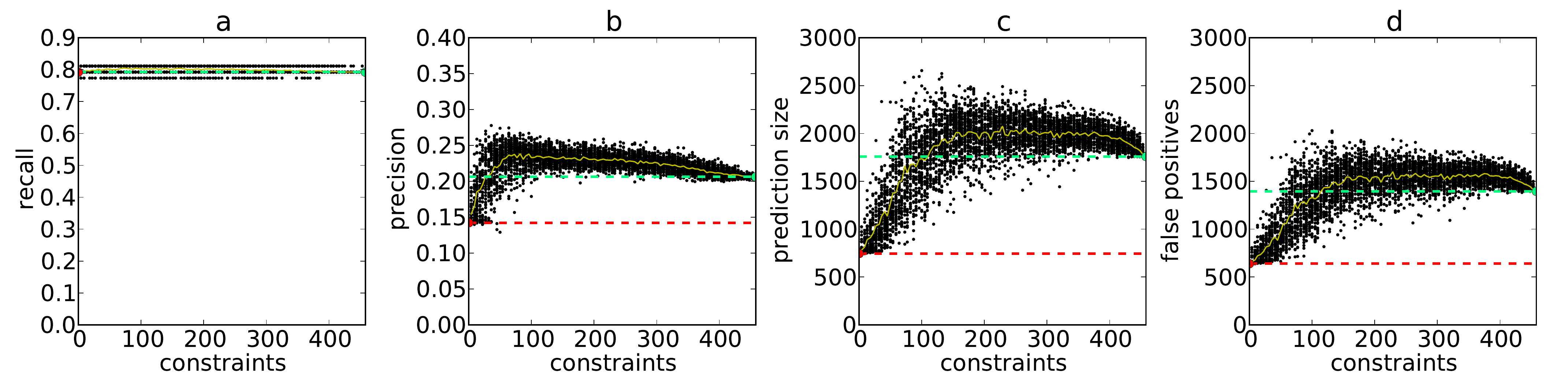}
\caption[Complex prediction quality as a function of the number of applied constraints]{
\label{complexpred:quality}
Complex prediction quality as a function of the number of applied constraints from 50 random samples for each step:
(a) recall, (b) precision, (c) total number of predicted complexes, (d) false positive predictions.
The red and green dashed lines mark the values obtained when none or all of the constraints are applied, respectively.
The yellow line indicates the mean value for each number of applied constraints.}
\end{figure*}

In the main article (Sec.~\ref{main:complex_prediction:experiments}), we showed that applying all 458 available constraints from \cite{jung2010} resulted in an improved precision, while leaving the recall constant when predicting the CYGD complexes.
Here we investigate the effect of a gradual application of constraints, in order to get an insight on their actual effects. 
Therefore, we randomly sampled subsets of all 458 available constraints of sizes between 4 ($1\%$ of 458) and 453 ($99\%$ of 458).
More precisely, for each $i\in\{1,2\dots,99\}$, we generated 50 independent samples of size $i\%$ of 458 (rounded to the nearest integer).

\paragraph{Precision and recall as a function of the number of applied constraints.}
Fig.~\ref{complexpred:quality} shows the development of precision and recall as a function of the number of applied constraints.
While the recall is independent of the number of applied constraints (Fig.~\ref{complexpred:quality}a),
high numbers of constraints ($\geq 100$) consistently provided an improvement in the precision over the unconstrained instance with precision $0.15$ (Fig.~\ref{complexpred:quality}b).
This indicates that constraining only about 1\% of the interactions is already sufficient to robustly improve complex prediction.
The maximum achieved precision then decreases gradually when applying more than 100 constraints and appears to reach a plateau upon using all available constraints.
The minimum achieved precision rarely drops below the final precision value $0.20$.

We offer the following explanation:
Note that initially both the number of predicted complexes and the number of false positive predictions increase (Figs.~\ref{complexpred:quality}c and~\ref{complexpred:quality}d), but the latter one at a slower rate.
Upon application of more interaction constraints both quantities reach a plateau and eventually decrease.
How might this come about?
A false positive complex may contain two interactions that are in reality mutually exclusive.
The corresponding constraint might not be sampled when applying few constraints, resulting in one false positive prediction.
When the number of constraints increases, the refinement step leads to two simultaneous protein subnetworks, on which again nearly the original complex without one of the exclusive interactions is predicted.
Each of the two complexes may now be closer to a true benchmark complex, but the number of constraints may still be too low to turn it into a true positive.
Thus refinement of one false positive complex might initially lead to two or more false positive smaller complexes.
This may underlie the observation that after an initial increase of the precision, showing a general beneficial effect of constraints, there is a stationary phase with even slightly decreasing precision.

Since the available constraints affect only less than 3\% of all interactions, an important challenge is to extrapolate the development of the precision as a function of much higher numbers of constraints.
It is rational to hypothesize that the precision will eventually start to increase upon constraining more interactions.
Since we cannot prove this hypothesis directly at this point, we instead mimicked the effect of adding further constraints that may destroy small false positive complexes ($< 24$ proteins) by artificially removing them from the prediction.
After an initial decrease this leads to a precision increase when applying more than 100 constraints (Fig.~\ref{figure:precision_big}).
This observation is consistent with our hypothesis that the application of further constraints should lead to an increase of precision on all complexes.

\begin{figure}[t]\centering%
\includegraphics[width=0.27\textwidth]{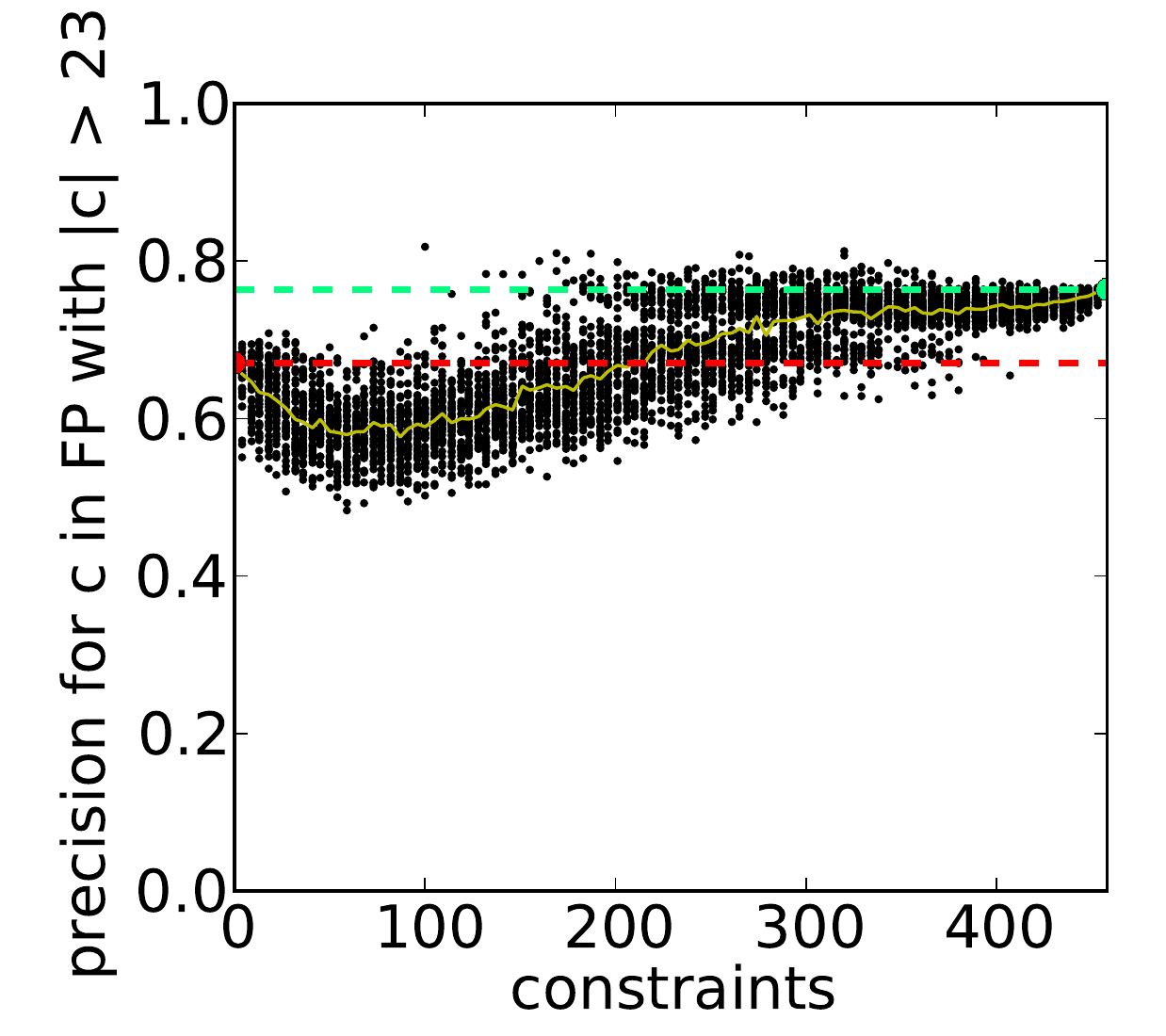}%
\vspace*{-3ex}%
\caption{Complex prediction precision when removing all false positive complexes that contain at most 23 proteins.
The red and green dashed lines mark the values obtained when none or all of the constraints are applied, respectively.
The yellow line indicates the mean value for each number of applied constraints.
}\label{figure:precision_big}%
\end{figure}

\paragraph{Accuracy on single complexes.}
\begin{figure}\centering
\includegraphics[width=0.7\columnwidth]{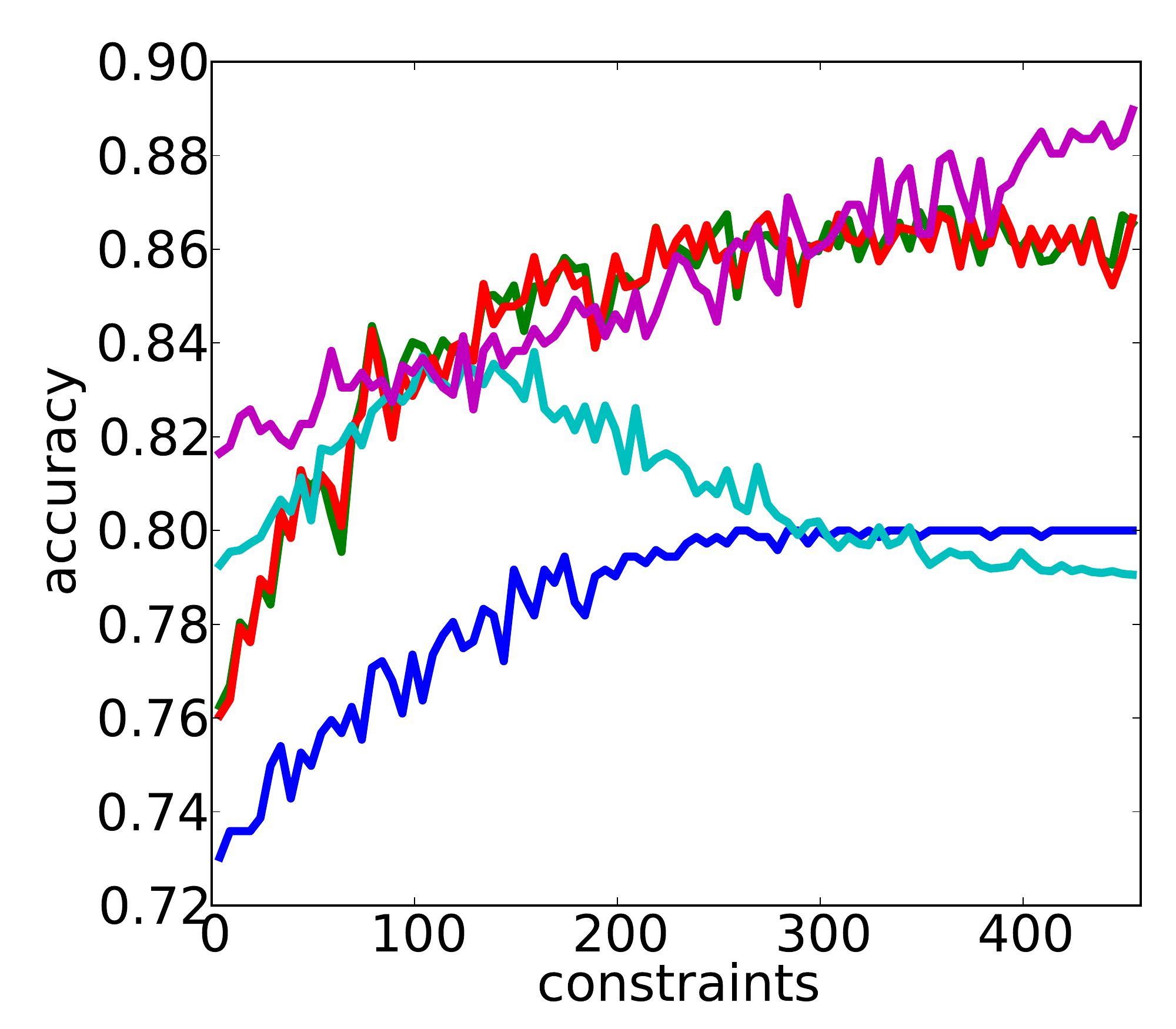}
\caption{\label{fig:single_complex_accuracy}
Matching accuracy for each CYGD complex over gradual application of constraints;
mean values over 50 independent samples. 
Complexes with constant accuracy and complex Nop58/Nop56/Nop1 are not shown.}
\end{figure}

Complementary to the precision and recall of the whole prediction, we examined single complexes as well.
As in Sec.~\ref{main:complex_prediction:experiments} of the main article,
we consider a predicted complex $c$ to match a CYGD complex $c'$ iff the matching accuracy
$\sqrt{ (|c \cap c'|^2) / (|c| \cdot |c'|) }$ exceeds the threshold of $\sqrt{0.2}$.

When monitoring the accuracy of each CYGD complex $c'$ while gradually introducing constraints,
we would expect that the accuracy with its best matching prediction $c$ remains constant or increases.
Indeed, while the accuracy remains constant for most of the 55 complexes, it increases for four complexes, but there are also two complexes whose accuracy does not follow this expectaton (Fig.~\ref{fig:single_complex_accuracy}; one of the latter ones not plotted).

The Nop58/Nop56/Nop1 complex (CYGD ID 440.12.30), one of the two complexes with decreasing accuracy, is not shown in Fig.~\ref{fig:single_complex_accuracy} because it contains one constraint (Nop56 and Nop58 are competing on the same binding domain of Nop1) so that it disappears once this constraint is applied.

The Gim3/Gim5/Gim4/PAC10/YKE2 complex (CYGD ID 177, cyan in Figure~\ref{fig:single_complex_accuracy}) accuracy first increases then decreases, approximately returning to the initial value in the end. 
We identified the following constraints to hurt its accuracy:
\begin{quote}
	\{SMC3,SMC3\}   $\Rightarrow \lnot$ \{SMC1, SMC3\} \\
	\{ARP6,SWD3\}  $\Rightarrow \lnot$ \{CLA4, SWD3\} \\
	\{SKP1,CDC53\} $\Rightarrow \lnot$ \{MET30, CDC53\}.
\end{quote}
These findings do not necessarily imply that those constraints are wrong. 
Rather they are hindering the heuristic LCMA in predicting the two complexes by altering the density of the corresponding regions in the simultaneous protein subnetwork. 
This shows that predicting complexes by density -- while it seems a good strategy in general -- does indeed fail for single cases.


\section{Generation of Random Constraints}
\label{sec:random_constraints}

It is important to compare the effects of (presumably) true known constraints with the effect of random constraints in order to show that observed effects are not simply due to applying constraints \textit{per se}.
Here, we specify how we generate random constraints of the type ``mutually exclusive interaction''.

To generate a random constraint, we randomly choose a protein $p_1 \in \proteins$ network, and randomly select two different neighbours $p_2, p_3 \in \proteins$. 
We interpret $p_1$ as the host protein and $p_2, p_3$ as two proteins competing on the same binding domain of $p_1$. 
Thereby we obtain the constraints $\{p_1, p_2\} \Rightarrow \lnot \{p_1, p_3\}$ and $\{p_1, p_3\} \Rightarrow \lnot \{p_1, p_2\}$. 

For the CYGD hypernetwork (Sec.~\ref{main:complex_prediction:experiments} of the paper), we iteratively generate 458 of these constraint pairs (we refer to such a pair simply as one constraint).
By independently repeating this process $n$ times, we gain $n$ independent samples of 458 random constraints.


\section{Protein Functional Importance Prediction with Hypernetworks}
\label{sec:pis}
In this section, we provide more details on protein complex prediction, supplementing Section~\ref{main:sec:functionalimportance} of the main article.

\subsection{Phenotypes of Null Mutants in the Saccharomyces Genome Database}
\label{sec:sgd}

We used the Saccharomyces Genome Database (SGD) \citep{sgd} to classify perturbations as lethal/sick. 
SGD collects the generated phenotypes of perturbation experiments for most proteins that are also in the CYGD, our selected benchmark. 
SGD phenotypes are provided in a standardized way, a complete list is provided by the \emph{Ontology Lookup Service} \citep{ols}.
To fit our modelling of perturbation and the notion of functional importance in the main article, we considered only ``null mutant'' perturbations rather than e.g. overexpression experiments. 
Table~\ref{table:lethalsick} shows those phenotypes that were counted as lethal/sick.
In contrast, the class of viable perturbations contains all that are annotated with the phenotype ``viable'' and are not contained in the class of lethal/sick ones.

\begin{table}\centering
\caption{SGD phenotypes selected to be classified as lethal/sick. 
A phenotype is composed by an observable and a qualifier.} 
\label{table:lethalsick}
\begin{tabular}{rl}
\toprule
observable & qualifier\\
\midrule
cell death & increased, increased rate \\
apoptosis& increased, increased rate \\
autolysis& increased, increased rate \\
cell lysis& increased, increased rate \\
necrotic cell death& increased, increased rate \\
competetive fitness & decreased \\
viability & decreased\\
vegetative growth & decreased\\
inviable & -\\
\bottomrule
\end{tabular}
\end{table}

\subsection{Analysis of cumulative PIS distributions}
\label{sec:cdf_analysis}
To find out about the capability of PIS to indicate functional important proteins, we analysed its distribution for disjoint classes like lethal/sick and viable. 
Therefore, we calculated empirical cumulative distribution functions (cdfs; normed cumulative histograms) for both classes. 
Maximum separation is provided if the ''lethal/sick`` cdf does not increase over 0 before the ''viable`` histrogram reaches 1 (Fig.~\ref{fig:maxsep}a). 
No separation means that both cdfs have the same values for each score. 
To better compare several cdf pairs, we calculate the absolute difference between the two cdfs (Fig.~\ref{fig:maxsep}b). 
The higher the absolute difference, the better is the separation. 
If the absolute difference reaches 1.0 at any point, the distributions are completely separated.

\begin{figure}[t]\centering%
\includegraphics[width=\columnwidth]{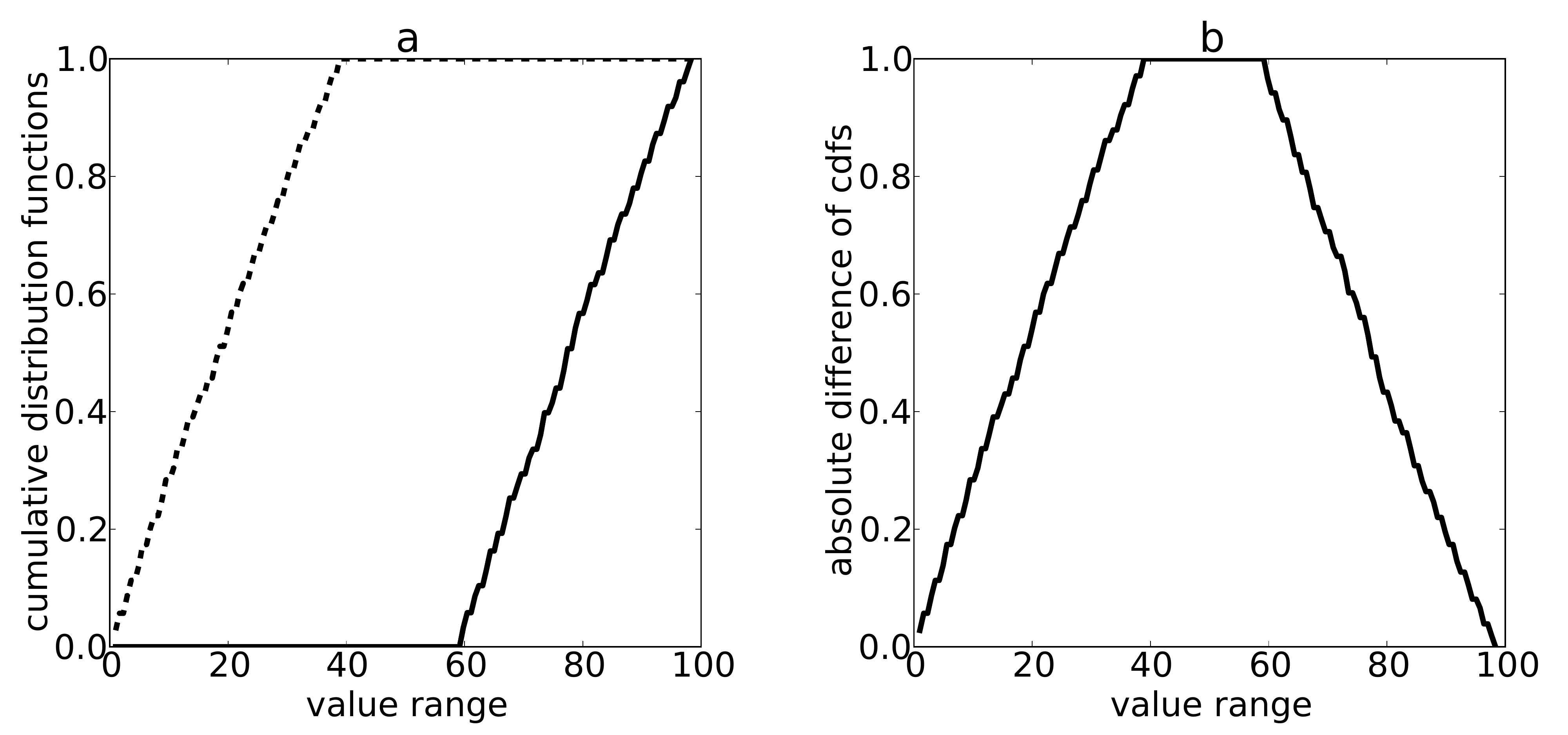}%
\caption{Example for cdf based separation analysis of two completely separated distributions: uniform distribution $U_1$ on $[0,40]$ and uniform distribution $U_2$ on $[60,100]$.
(a) cumulative distribution functions (cdfs) for $U_1$ (dashed) and $U_2$ (solid). 
x-axis: value range of the distributions. E.g. at $x=50$, the cdf of $U_1$ has reached $1$, while that of $U_2$ is still zero.
(b) absolute difference of the two cdfs. 
}\label{fig:maxsep}
\end{figure}


\begin{table*}
\centering
\caption{Representation of possible interaction constraints in structured text. Keywords and proteins are annotated with a defined ontology term id. This way, finding and parsing constraints is improved while human readability is maintained.} \label{table:structured_text}
\begin{tabular}{rp{11cm}}
\toprule
constraint & structured text\\
\midrule
mutually exclusive interactions & x (uniprotkb:x) competes (MI:0941) with y (uniprotkb:y) for interaction (MI:0407) with z (uniprotkb:z).\\
negative allosteric regulation by protein binding & interaction (MI:0407) between x (uniprotkb:x) and y (uniprotkb:y) allosterically (SBO:0000239) inhibits (SBO:0000407) the interaction (MI:0407) of y (uniprotkb:y) with z (uniprotkb:z).\\
positive allosteric regulation by protein binding & interaction (MI:0407) between x (uniprotkb:x) and y (uniprotkb:y) allosterically (SBO:0000239) activates (SBO:0000461) the interaction (MI:0407) of y (uniprotkb:y) with z (uniprotkb: z).\\
negative regulation by phosphorylation & interaction (MI:0407) between x (uniprotkb:x) and  y (uniprotkb:y) is inhibited (SBO:0000407) if x (uniprotkb:x) is phosphorylated (GO:0016310) on residue $i$.\\
positive regulation by phosphorylation & interaction (MI:0407) between x (uniprotkb:x) and  y (uniprotkb:y) is activated (SBO:0000461) if x (uniprotkb:x) is phosphorylated (GO:0016310) on residue $i$.\\
\bottomrule
\end{tabular}
\end{table*}

\section{Software Implementation}
\label{sec:implementation}

We implemented protein hypernetworks as a JAVA\textsuperscript{TM} based software suite. 
The suite consists of \emph{ProteinHypernetworkEditor} that allows the definition and editing of protein hypernetworks, and \emph{ProteinHypernetwork} that imple\-ments the prediction methods presented in the main article.
Further, both tools provide a graphical user interface and extensive visualization and import/export capabilities. 
The software suite can be obtained at {\href{http://www.rahmannlab.de/research/hypernetworks}{http://www.rahmannlab.de/research/hypernetworks}}.

\subsection{Representation of Constraints}
\label{section:representation}
An important challenge is the definition of a widely accepted format for the interchange of interaction dependencies or constraints.
While SBML \citep{sbml} is suitable in principle, it provides a biochemical view of interactions and therefore contains overhead that is unnecessary for the definition of a protein hypernetwork.
Instead we propose a two level approach for the interchange of constraints.
\paragraph{Level 1: Structured Text.}
\citet{ceol2008} proposed a machine readable structured abstract that should be published along with papers on protein interactions to allow automated curation (see also \citep{Leitner:2010fk}).
The format combines human-readable sentences with machine readable annotation, and is already capable of representing the expected types of constraints.
For example, a pair of mutual exclusive interactions (as reported by \citet{jung2010}) can be represented as follows:
''ARC40 (uniprotkb:P38328) is a \emph{competitor (MI:0941)} of BEM2 (uniprotkb:P39960) for \emph{interaction (MI:0317)} with CLA4 (uniprotkb:P48562)``.
Table \ref{table:structured_text} provides a generalized representation for the major types of interaction constraints.

\paragraph{Level 2: HypernetworkML.} \mbox{ }
With the hypernetwork markup language (HypernetworkML, \citet{hypernetworkml}) we provide an XML-based file format \citep{xml} that is more suitable for possible large-scale studies providing many constraints at once and for permanent storage of the data.
HypernetworkML is a combination of two established XML based formats:
Interactions and proteins are represented as a graph using GraphML \citep{graphml} whereas embedded MathML \citep{mathml} is used for a propositional logic definition of constraints.
Hence, HypernetworkML is capable to represent a complete protein hypernetwork while maintaining compatibility with known standards.

\subsection{Resource Consumption}
A complex prediction on the defined yeast hypernetwork takes 12 seconds using all 4 cores of an Intel\textsuperscript{\textregistered} Core\textsuperscript{TM} i5 CPU with 2.8GHz.
The prediction of the PIS of 4579 single protein perturbations takes 8 seconds.
In both cases the software uses approximately 750MB RAM during prediction.

\end{document}